\documentclass[runningheads]{llncs}
\usepackage{subfigure}
\usepackage[textsize=footnotesize]{todonotes}

\usepackage{amsthm,amssymb}
\usepackage{xspace}

\usepackage[modulo]{lineno}

\usepackage{version}
\newif\iffull
\fulltrue

\graphicspath{{figures/}}

\newtheorem{observation}{Observation}

\newcommand{\propertwo}{fan-planar proper $2$-layer\xspace}

\newcommand{\calC}{\mathcal{C}}
\newcommand{\calP}{\mathcal{P}}
\newcommand{\stego}{stegosaurus\xspace}
\newcommand{\DP}{{\sc PlanarDP}\xspace}

\newcommand{\Dujmovic}{Dujmovi\'{c}\xspace}

\usepackage[misc]{ifsym}

\newcommand{\lay}[1]{${#1}$-layer}
\newcommand{\ld}[1]{\lay{#1} drawing\xspace}
\newcommand{\lds}[1]{\lay{#1} drawings\xspace}

\begin{document}

\title{On 
Layered Fan-Planar Graph Drawings\thanks{Research of TB supported by NSERC. 
Research of SC supported by DFG grant WO 758/11-1.
Research of  FM supported in part by MIUR under Grant 20174LF3T8 AHeAD: efficient Algorithms for HArnessing networked Data.   
Research initiated while FM was visiting the University of Waterloo and continued at the Bertinoro Workshop on Graph Drawing 2018.}}
    
\author{
Therese Biedl\inst{1}$^{\mbox{\Letter}}$ \and
Steven Chaplick\inst{2} \and
Ji\v{r}i Fiala \inst{3} \and
Michael Kaufmann \inst{4} \and
Fabrizio~Montecchiani\inst{5} \and
Martin N\"{o}llenburg\inst{6} \and
Chrysanthi Raftopoulou \inst{7}
}
\authorrunning{Biedl, Chaplick, Fiala, Kaufmann, Montecchiani, N\"{ollenburg}, Raftopolou}
\institute{
University of Waterloo, Canada, \email{biedl@uwaterloo.ca}
\and
Universit\"{a}t W\"{u}rzburg, Germany 
\and
Charles University, Prague, Czech Republic
\and
Universit\"{a}t T\"{u}bingen, Germany
\and
Universit{\`a} degli Studi di Perugia, Italy
\and
Algorithms and Complexity Group, TU Wien, Vienna, Austria 
\and
National Technical University of Athens, Greece 
}

\maketitle

\begin{abstract}
In this paper, we 
study fan-planar drawings that use $h$ layers and are proper, i.e.,
edges connect adjacent layers.
We show that if the embedding of
the graph is fixed, then testing the existence of such drawings
is fixed-parameter tractable in $h$, via a reduction to a similar result for planar
graphs by \Dujmovic et al.  If the embedding is not fixed, then we give partial
results for $h=2$: It was already known how to test existence of fan-planar proper \lds{2} for 2-connected graphs, 
and we show here how to test this for trees.  Along the
way, we exhibit other interesting results for graphs with a fan-planar proper \ld{h};
in particular we bound their pathwidth and show
that they have a bar-1-visibility representation.
\end{abstract}

\section{Introduction}

In a seminal paper, \Dujmovic, Fellows, Kitching, Liotta, McCartin,
Nishimura, Ragde, Rosamond, Whitesides and Wood showed that testing
whether a planar graph has a proper layered drawing of height $h$
is fixed-parameter tractable in $h$ \cite{DFK+08}.  (Detailed definitions
are in the next section.)  This is of interest
since finding a proper layered drawing of minimum height is NP-hard
\cite{HR92}.  \Dujmovic et al.~also study
some variations, such as having a constant number of crossings
or permitting flat edges and long edges.

In this paper, we aim to generalize their results to graphs that are
{\em near-planar}, i.e., graphs that may have crossings, but
there are restrictions on how such crossings may occur.  Such
graphs have been the object of great interest in the graph drawing
community in recent years (refer to~\cite{DBLP:journals/csur/DidimoLM19,DBLP:journals/csr/KobourovLM17} for surveys).  We study {\em 1-planar graphs}
\iffull
 where every edge has at most one crossing, 
\fi
 and {\em fan-planar graphs}%
\iffull
 where an edge $e$ may have many crossings, but all the edges crossed
by $e$ must have a common endpoint.
Every 1-planar graph is also fan-planar.

\fi
Our main result is that for a fan-planar graph $G$ with a fixed embedding,
we can test in time fixed-parameter tractable in $h$ whether $G$ has a proper layered drawing on $h$ layers that respects the embedding.  
Our approach is to 
\iffull
reduce the problem
to the existence of a proper planar $f(h)$-layer drawing for some
suitable function $f(h)\in \Theta(h)$, i.e., we 
\fi
modify $G$ to obtain
a planar graph $G'$ that has a planar \ld{f(h)} if and only if $G$
has a fan-plane \ld{h}.  We then appeal to the result by \Dujmovic et al.
Nearly the same approach also works for short drawings where flat edges
are allowed, and for 1-planar graphs it also works for long edges when
drawn as $y$-monotone polylines.
\iffull
(In contrast to planar drawings, such 1-planar drawings cannot
always be ``straightened'' into a straight-line drawing.)
\fi

The above algorithms crucially rely on the given embedding. 
\iffull
We also study the case where the embedding can be chosen.  
Here it was
known how to test whether the graph has a proper drawing on 2 layers
if the graph is 2-connected~\cite{DBLP:journals/jgaa/BinucciCDGKKMT17}, with the main
ingredient that the structure of such graphs can be characterized.
To push this towards an algorithm for all graphs, we study the following
problem:  Given a tree $T$, does it have a fan-planar proper drawing
on 2 layers?  We give a dynamic programming (DP) algorithm that answers
this question in linear time.   The algorithm is not at all the
usual straightforward bottom-up-approach; instead we need to analyze
the structure of a tree with a fan-planar proper \ld{2} carefully.  
\fi

One crucial ingredient for the algorithm by \Dujmovic et al.~\cite{DFK+08}
is that a graph with a planar proper \ld{h} has pathwidth at most $h-1$,
and this bound is tight.
We similarly can bound the pathwidth for graphs that have a fan-planar proper \ld{h},
and again the bound is tight.
The proof uses a detour: we show that 
graphs with a fan-planar proper layered drawing have a {\em bar-$1$-visibility
representation}, a result of interest in its own right.

The paper is organized as follows.  After reviewing definitions,
we start with the result about bar-1-visibility representations and the 
pathwidth, since these are convenient warm-ups for dealing with 
fan-planar proper layered drawings.  We then give the reduction from
fan-plane proper \ld{h} to planar proper \ld{f(h)} and hence prove
fixed-parameter tractability of the existence of fan-plane proper \ld{h}.
Finally we turn towards fan-planar proper 2-layer drawings, and show how
to test the existence of such drawings for trees in linear time.  
All our algorithms are constructive, i.e., give such drawings in case
of a positive answer.
We conclude with open problems.

\section{Preliminaries}\label{se:preliminaries}

We assume familiarity with graphs and graph terminology. Let $G=(V,E)$
be a graph.  We assume throughout that
$G$ is connected and simple.

A \emph{path decomposition} $P$ of a graph $G$ is a sequence 
$P_1,\dots,P_p$ of vertex sets (``\emph{bags}'') that satisfies:
(1) every vertex is in at least one bag, (2) for every edge $(v,w)$
at least one bag contains both $v$ and $w$, and (3) for every vertex
$v$ the bags containing $v$ are contiguous in the sequence.
The \emph{width} of a path decomposition is $\max\{|P_t|-1: 1 \leq t \leq p\}$. The \emph{pathwidth} $pw(G)$ of a graph $G$ is the minimum width of any path decomposition of $G$.

\newcommand{\fanplanarfootnote}{\footnote{There are further restrictions, see e.g.~\cite{KaufmannU14}. These are automatically
satisfied if the graph has a proper layered drawing and so will not be reviewed here.}  }

\paragraph{\bf Embeddings and drawings that respect them:}
\iffull
We mostly follow the notations in \cite{Schaefer18}.
Let $\Gamma$ be a {\em drawing} of $G$, i.e.,
an assignment of distinct points to vertices and non-self-intersecting curves
connecting the endpoints to each edge.  
All drawings are assumed to be {\em good}:  
No edge-curve intersects a vertex-point unless it is its endpoint, 
no three edge-curves intersect in one point,
any two edge-curves intersect each other in at most one point (including
a shared endpoint),
and any two edge-curves that intersect do so while crossing transversally
(and we call this point a {\em crossing}).
An {\em edge-segment} is a maximal (open) subset of an edge-curve that contains 
no crossing or vertex-point.
In what follows, we usually identify the graph-theoretic object
(vertex, edge) with the geometric object (point, curve) that represents it.

The {\em rotation} at a vertex $v$ in the drawing is the cyclic order in which
the incident edges end at $v$.
(Often we list the neighbours rather than the edges.)
The {\em rotation system} of a drawing consists of the set of 
rotations at all vertices.
A {\em region} of a drawing $\Gamma$ is a maximal connected part of 
$\mathbb{R}^2\setminus \Gamma$; it can be identified by listing the
edge-segments, crossings and vertices on it in clockwise order.  
The {\em planarization} of a drawing is obtained by replacing every crossing
by a new vertex of degree 4 (called a {\em (crossing)-dummy-vertex}).

A graph is called {\em $k$-planar} (or simply {\em planar} for $k{=}0$)
if it has a {\em $k$-planar drawing} where every edge has at most $k$ crossings.
In a planar drawing the regions are called
{\em faces} and the infinite region is called the {\em outer-face}.
A drawing of $G$ is called {\em fan-planar}  if it has a {\em fan-planar drawing}
where for any edge $e$, all edges $e_1,\dots,e_d$ that are crossed by
$e$ have a common endpoint $v$.%
\fanplanarfootnote
The set $\{e_1,\dots,e_d\}$ is also called a {\em fan} with {\em center-vertex} $v$.

A {\em planar embedding} of a graph $G$ 
consists of the rotation system obtained from some planar drawing of $G$
as well as a specification of outer-face.
An {\em (abstract) embedding} of a graph $G$ consists of a graph $G_P$ with
a planar embedding that is the planarization of some drawing of $G$.
Put differently, an embedding of $G$ specifies the rotation system,
the pairs of edges that cross, the order in which the crossings occur
along each edge, and the infinite region.   A drawing of $G$ is called
{\em embedding-preserving} if its planarization is $G_P$.  We use
{\em plane/1-plane/fan-plane} for a graph $G$ together with an abstract embedding
corresponding to a planar/1-planar/fan-planar drawing, and also for an embedding-preserving 
drawing of $G$.
\fi

\paragraph{\bf Layered drawings:}

Let $h\geq 1$ be an integer. 
An \emph{\ld{h}} of a graph $G$ is
a drawing where the vertices are on one of $h$ distinct horizontal lines $L_1, \dots, L_h$, called \emph{layers}, and edges are drawn as $y$-monotone polylines
for which all bends lie on layers.  We enumerate the layers top-to-bottom.

\begin{figure}[tb]
\hspace*{\fill}%
\includegraphics[width=0.49\linewidth,page=2]{stego.pdf}
\hspace*{\fill}%
\includegraphics[width=0.49\linewidth,page=1]{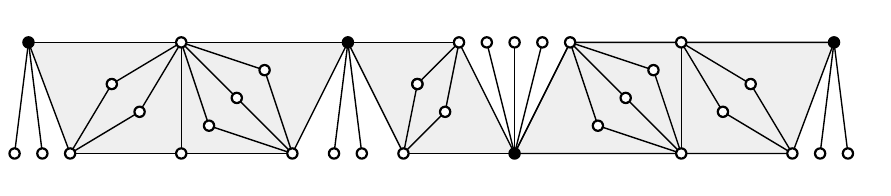}
\hspace*{\fill}%
\caption{A \propertwo drawing and its graph (a stegosaurus).}
\label{fig:stegoEx}
\end{figure}

Layered drawings are further distinguished by what types of edges are allowed;
the following notation is from \cite{Sud04}.
An edge is called {\em flat} if its endpoints lie on the same layer,
{\em proper} if its endpoints lie on adjacent layers, and {\em long}
otherwise.  A {\em proper} \ld{h} contains only proper edges, 
a {\em short} \ld{h} contains no long edges, an {\em upright} \ld{h}
contains no flat edges, and an {\em unconstrained} \ld{h} permits 
any type of edge.%
\footnote{The terminology is slightly different in the paper by \Dujmovic
et al. \cite{DFK+08}; for them {\em any} \ld{h} was required to be short.}
Any graph with a planar upright \ld{h} has pathwidth at most $h$, and at most
$h{-}1$ if there are no flat edges \cite{DFK+08,FLW03}.
Any graph with a \propertwo drawing is a subgraph of a so-called \stego
(see Fig.~\ref{fig:stegoEx} and Section~\ref{sec:trees}) \cite{DBLP:journals/jgaa/BinucciCDGKKMT17};
those have pathwidth 2.

A key concept for us is where crossings can be in proper layered drawings and how to group them.
Let $G_P$ be the planarization of some graph $G$ with a fixed embedding.
As in Fig.~\ref{fig:patch}, a {\em crossing-patch} $\calC$
is a maximal connected subgraph of $G_P$ for which all
vertices are crossing-dummy-vertices.  
Let $E_{\calC}$ be the edges of $G$ that have crossings in $\calC$,
let $V_\calC$ be the endpoints of $E_{\calC}$, and let $G_C$ be the graph $(V_C,E_C)$.
Since any edge connects two adjacent layers, and a crossing-patch is
connected, we can observe:

\begin{observation}
\label{obs:crossing_patch}
If $G$ has a proper embedding-preserving layered drawing $\Gamma$ 
then all crossings of a crossing-patch $\calC$ lie strictly between 
two consecutive layers, and the vertices in $V_{\calC}$ lie on those layers.
\end{observation}

\begin{figure}[tb]
\hspace*{\fill}%
\subfigure[~]{\includegraphics[width=0.26\linewidth,page=2]{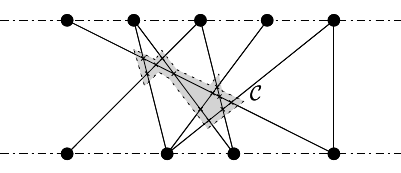} }%
\hspace*{\fill}%
\subfigure[~]{\includegraphics[width=0.26\linewidth,page=3]{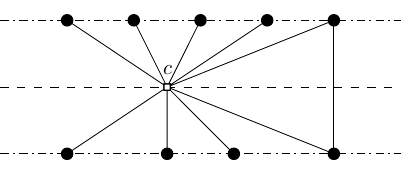} }%
\hspace*{\fill}%
\caption{A crossing-patch in a graph that is not fan-planar, and how to contract it.}
\label{fig:sample}
\label{fig:replace_crossings}
\label{fig:patch}
\label{fig:weakly_isomorphic}
\end{figure}

\section{Bar-Visibility representations and Pathwidth} 

In this section, we show that a graph
with a fan-planar short \ld{h} has pathwidth at most $2h-1$ (and at most
$2h-2$ if the drawing is proper).
The proof uses a
{\em bar-$c$ visibility representation}, which is an assignment of
a horizontal line segment ({\em bar}) to every vertex and a vertical line segments connecting
the bars of endpoints to every edge in such a way that bars are disjoint
and every edge-segment contains at most $c$ points (not counting the 
endpoints) that belong to bars.

\begin{figure}[tb]
\hspace*{\fill}
\subfigure[~]{\includegraphics[width=0.23\linewidth,page=2]{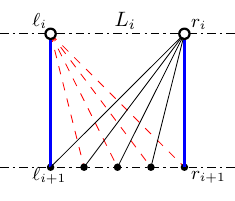}
\label{fig:bar1vis_top}}
\hspace*{-5mm}
\subfigure[~]{\includegraphics[width=0.23\linewidth,page=3]{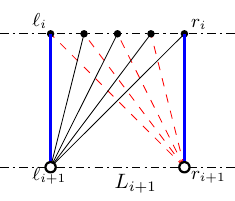}
\label{fig:bar1vis_bottom}}
\hspace*{-5mm}
\subfigure[~]{\includegraphics[width=0.23\linewidth,page=4]{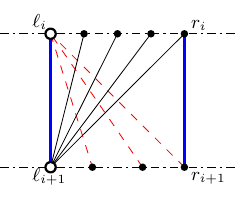}
\label{fig:bar1vis_left}}
\hspace*{-5mm}
\subfigure[~]{\includegraphics[width=0.23\linewidth,page=5]{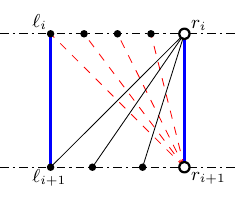}
\label{fig:bar1vis_right}}
\hspace*{\fill}
\caption{A fan-planar proper \ld{2}; planar edges that separate fan-subgraphs are
blue (thick).  [For labelling-purposes we show the planar edges as
vertex-disjoint, but consecutive ones could have vertices in common.]
We show the four possible locations of center-vertices (white). 
}
\label{fig:bar1vis_categories}
\end{figure}

\begin{theorem}
\label{thm:bar1}
If $G$ has fan-planar proper \ld{h} $\Gamma$, then $G$ has a bar-1-visibility
representation.  Moreover, any vertical line intersects at most
$2h-1$ bars of the visibility representation.
\end{theorem}

\begin{proof}
In the first step, make $\Gamma$ maximal, i.e., insert all edges that can be added 
while keeping a fan-planar proper \ld{h}.  In the resulting drawing 
every crossing-patch is enclosed
by two planar edges (shown thick blue in 
Fig.~\ref{fig:bar1vis_categories}).  The subgraph between two
such planar edges consists (if it has crossings at all) of two crossing fans; we 
call this a {\em fan-subgraph}.  Studying all possible positions of these two fans,
we see that the two center-vertices include exactly one of the top vertex of the left planar edge
or the bottom vertex of the right planar edge.    
We remove the crossed edges incident to this center-vertex in the
fan-subgraph; see Fig.~\ref{fig:bar1vis_categories} where removed edges are red (dashed).
The remaining graph $G'$ is planar and has a planar proper \ld{h}.
We can convert this into a bar-0-visibility representation $\Gamma'$ where the
layer-assignment and the order within layers is unchanged \cite{Bie14},
in particular any vertical line intersects at most $h$ vertex-bars.

Next, shift bars upward until bars of each layer lie ``diagonally'', see the dark gray bars in Fig.~\ref{fig:bar1vis_categories_bars}.
More precisely, we process layers from bottom to top.
For each layer we assign increasing $y$-coordinates to the bars from left to right such that every bar has its own $y$-coordinate.

\begin{figure}[tb]
\hspace*{\fill}
\subfigure[~]{\includegraphics[width=0.2\linewidth,page=7]{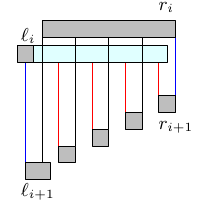}%
\label{fig:bar1vis_top_bar}}%
\hspace*{\fill}
\subfigure[~]{\includegraphics[width=0.2\linewidth,page=8]{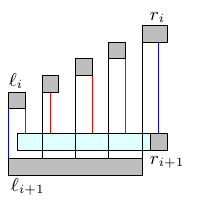}%
\label{fig:bar1vis_bottom_bar}}%
\hspace*{\fill}
\subfigure[~]{\includegraphics[width=0.2\linewidth,page=9]{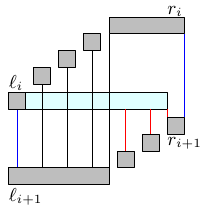}%
\label{fig:bar1vis_left_bar}}%
\hspace*{\fill}
\subfigure[~]{\includegraphics[width=0.2\linewidth,page=10]{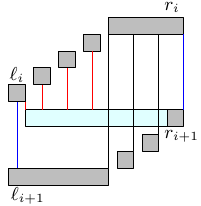}%
\label{fig:bar1vis_right_bar}}%
\hspace*{\fill}
\caption{Bar-1-visibility representation for different types of fan-subgraphs.}
\label{fig:bar1vis_categories_bars}
\end{figure}

Let the planar edges to the left and right of a fan-subgraph be $(\ell_i,\ell_{i+1})$ and $(r_i,r_{i+1})$, with vertices
indexed by layer.
The process of removing edges ensures that all of the missing edges are incident to $r_{i+1}$ or $\ell_i$.
If they were incident to $\ell_i$, then
we extend $\ell_i$ to the right until it vertically sees its diagonally opposite corner $r_{i+1}$. 
Otherwise, we extend $r_{i+1}$ to the left until it vertically sees its diagonally opposite corner $\ell_i$.
This extension realizes all removed edges of the fan-subgraph, since the extended bar can see vertically all other bars of 
vertices of the fan-subgraph.  
By our construction, the extended bars do not cross the planar edges between $\ell_i$ and $\ell_{i+1}$, or between $r_i$ and $r_{i+1}$. 
Since for each fan-subgraph there is only one extended bar, the edges of $G$ that belong to $G'$ go through at most one extended bar. 
Therefore the computed representation is a bar-1-visibility representation of $G$. 
In each fan-subgraph only one bar is extended, therefore every vertical
line intersects at most $h$ bars from the $h$ layers and at most $h-1$
bars from the $h-1$ fan-subgraphs that it traverses.
\end{proof}

\iffull
With a minor change, we can prove a similar result for short layered drawings.
\fi

\begin{theorem}
\label{thm:bar1_short}
If $G$ has a fan-planar short \ld{h} $\Gamma$, then $G$ has a bar-1-visibility
representation where any vertical line intersects at most $2h$ bars of
the visibility representation.
\end{theorem}
\iffull
\begin{proof}
Let $G^-$ be the graph obtained by removing all flat edges; this has a
fan-planar proper \ld{h} and therefore a bar-1-visibility representation 
using Theorem~\ref{thm:bar1}.    Let $\Gamma'$ be the visibility representation
(of some subgraph of $G^-$) used as intermediate step in this proof.
Lengthen the bars of $\Gamma'$ maximally so that within any layer, the bar
of one vertex $v$ ends exactly where the bar of the next vertex $w$ begins.  
(Note that no vertical edge-segment 
lies between the bars of $v$ and $w$ since there are no
long edges.)  We have some choice in how much to extend $v$ vs.~how much to
extend $w$ into the gap between them, and do this such that no two points 
where bars begin/end have the same $x$-coordinate.

Now convert this visibility representation into a bar-1-visibility 
representation $\Gamma^-$ of $G^-$ exactly as before.  We claim that this is the
desired bar-1-visibility representation of $G$.  Consider a flat edge
$e=(v,w)$, with (say) $v$ left of $w$ on their common layer.  Let $X_e$ be
the $x$-coordinate where the bar of $v$ ends and the bar of $w$ begins in
the modified $\Gamma'$.    To obtain $\Gamma^-$, these bars are first shifted
to different $y$-coordinates (without changing $x$-coordinates of their
endpoints). Since $v$ and $w$ are consecutive
within one layer of $\Gamma'$, they end of on consecutive layers of
$\Gamma^-$.  Next the bars are (possibly) lengthened, but never shortened.
Therefore edge $(v,w)$ can be inserted with $x$-coordinate $X_e$ to connect
the bars of $v$ and $w$.

It was argued in Theorem~\ref{thm:bar1} that any vertical line intersects at
most $2h-1$ bars in that construction.  The only change in our construction
is that sometimes endpoints of bars may have the {\em same} $x$-coordinate $X_e$
(for some flat edge $e$), which means that the vertical line with $x$-coordinate
$X_e$ now may intersect more bars.
However, we ensured that $X_e\neq X_{e'}$ for any two flat edges $e,e'$,
which means that even at $x$-coordinate $X_e$ the vertical line intersects
at most $2h$ bars.
\end{proof}
\fi

\begin{corollary}
\label{cor:pw}
\label{cor:proper_to_pw}
If $G$ has a fan-planar proper \ld{h}, then $pw(G)\leq 2h-2$.
If $G$ has a fan-planar short \ld{h}, then $pw(G)\leq 2h-1$.
\end{corollary}
\begin{proof}
Take the bar-1-visibility representation of $G$ from Theorem~\ref{thm:bar1}
[respectively \ref{thm:bar1_short}]
and read a path decomposition $\calP$ from it.  To do so, sweep a vertical
line $\ell$ from left to right.  Whenever $\ell$ reaches the $x$-coordinate
of an edge-segment, attach a new bag $P$ at the right end of $\calP$ and insert
all vertices that are intersected by~$\ell$.
The properties
of a path decomposition are easily verified since bars span a contiguous set of $x$-coordinates,
and for every edge $(v,w)$ the line through the edge-segment intersects both
bars of $v$ and $w$.  Since any vertical line intersects at most $2h-1$
[$2h$, respectively] bars, each
bag has size at most $2h-1$ [$2h$] and the width of the decomposition is at most $2h-2$
[$2h-1$].
\end{proof}


We now show that the bounds of Corollary~\ref{cor:pw} are tight, even for trees.

\begin{theorem}
For any $h\geq 1$, there are trees $T^p_{2h{-}2}$ and $T^s_{2h{-}1}$ such that
\begin{itemize}
\item $T^p_{2h{-}2}$ has a fan-planar proper \ld{h} and $pw(T^p_{2h{-}2})\geq2h{-}2$, 
\item $T^s_{2h{-}1}$ has a fan-planar short \ld{h} and $pw(T^s_{2h{-}1})\geq 2h{-}1$.
\end{itemize}
\end{theorem}
\iffull
\begin{proof}
Roughly speaking, for $\alpha\in \{s,p\}$, $T_i^{\alpha}$ is the complete ternary tree with some (but not all) edges subdivided.
To be more precise, for $h=1$, define $T_0^p$ to be a single node $r_0$, which can drawn on one layer and has pathwidth $0=2h-2$.  Define $T_1^s$ to be an edge $(r_1,\ell)$, which can be drawn as a flat edge on one layer and has pathwidth $1=2h-1$.

For $\alpha\in \{s,p\}$ and any $i$ where $T_i^\alpha$ is not yet defined,
set $\overline{T}_i^{\alpha}$ to be a new
vertex $r_i$ with three children, and make each child a root of $T_{i-1}^\alpha$.  Clearly $pw(\overline{T}_i^\alpha)\geq pw(T_{i-1}^\alpha)+1$, since removing $r_i$ from $\overline{T}_i^\alpha$ gives three components that each contain $T_{i-1}^\alpha$.    
To obtain $T_i^\alpha$ from $\overline{T}_i^\alpha$ we subdivide some edges (see below).
This cannot decrease the pathwidth, so using induction one shows that $pw(T_i^\alpha)\geq i$.

Figure~\ref{fig:completeTernary} shows that for all $i$ where $\overline{T}^\alpha_{i-2}$ is defined,
$\overline{T}^\alpha_{i}$ has a fan-planar drawing
with one more layer than used by $\overline{T}^\alpha_{i-2}$.
Furthermore, $r_i$ is in the top row, and every edge is drawn properly,
presuming we subdivide the edges incident to $r_i$.  Using induction 
therefore $T_{2h-2}^p$ and $T_{2h-1}^s$ have fan-planar \lds{h}.
\end{proof}
\fi

\begin{figure}[tb]
\subfigure[~]{\includegraphics[width=0.22\linewidth,page=9,trim=55 0 60 0,clip]{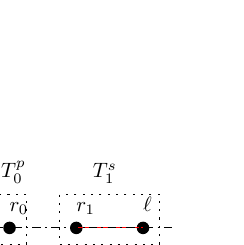}}
\hspace*{\fill}
\subfigure[~]{\includegraphics[width=0.32\linewidth,page=7,trim=30 0 30 0,clip]{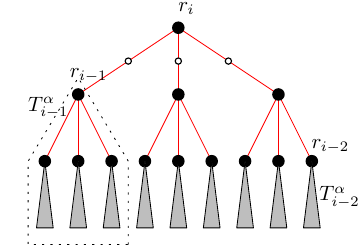}}
\hspace*{\fill}
\subfigure[~]{\includegraphics[width=0.42\linewidth,page=8]{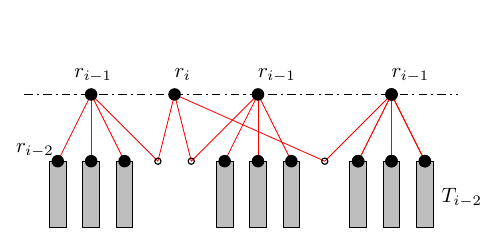}}
\caption{(a) The trees for $h=1$.  
(b) Constructing $T^\alpha_{i}$ from $T^\alpha_{i-2}$ and (c) drawing it
using one added layer.}
\label{fig:completeTernary}
\end{figure}

Note that the drawing in Fig~\ref{fig:completeTernary}(c) are fan-planar, but not 1-planar.
This naturally raises the question:  What is the pathwidth of a graph that
has a 1-planar \ld{h}?  We suspect that it cannot be more than $\approx \frac{3}{2}h$
(this remains open), and can show that for the above trees (subdivided differently)
this bound would be tight.  

\begin{theorem}
For any odd $h\geq 1$ (say $h=2k+1$ with $k\geq 0$), there are trees $T^p_{3k}$ and $T^s_{3k{+}1}$ such that 
\begin{itemize}
\item 
$T^p_{3k}$ has a 1-planar proper \ld{h} and $pw(T^p_{3k}){\geq}3k=\frac{3}{2}h-\frac{3}{2}$, and
\smallskip
\item 
$T^s_{3k{+}1}$ has a 1-planar short \ld{h} and $pw(T^s_{3k{+}1})\geq 3k{+}1=\frac{3}{2}h-\frac{1}{2}$.
\end{itemize}
\end{theorem}

\iffull
\begin{proof}
Define $T_0^p$ and $T_1^s$ exactly as in the previous proof; their drawings have no crossings.
Also define $\overline{T}^\alpha_i$ as before, but subdivide edges differently to obtain $T^\alpha_i$; 
see below.  
Figure~\ref{fig:completeTernary1Planar} shows that
$\overline{T}^\alpha_{i}$ has a 1-planar drawing
with two more layers than 
$\overline{T}^\alpha_{i-3}$ (for all $i$ where $\overline{T}^\alpha_{i-3}$ is defined).
Furthermore, $r_i$ is in the top row, and every edge is drawn properly,
presuming we subdivide  two edges incident to $r_i$ and all child-edges at the
child $r_{i-1}$ whose parent-edge was not subdivided.  The result now follows
using induction on $h$.
\end{proof}
\fi

\begin{figure}[tb]
\hspace*{\fill}
\subfigure[~]{\includegraphics[width=0.4\linewidth,page=7]{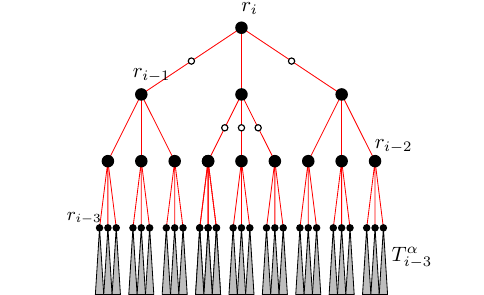}}
\hspace*{\fill}
\subfigure[~]{\includegraphics[width=0.4\linewidth,page=8]{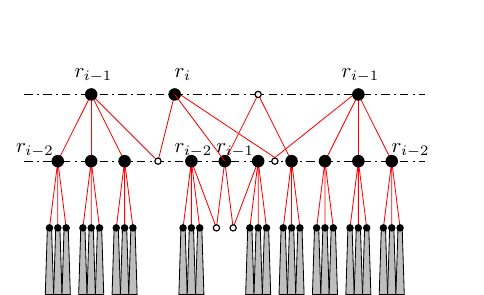}}
\hspace*{\fill}
\caption{(a) Constructing $T^\alpha_{i}$ from $T^\alpha_{i-3}$ and (b) drawing it
using two added layers.}
\label{fig:completeTernary1Planar}
\end{figure}

\section{Testing Algorithm for Embedded Graphs}

This section presents FPT-algorithms to determine whether an 
embedded graph~$G$ has an embedding-preserving \ld{h}.
The first algorithm tests the existence of a {\em proper} drawing,
and can be applied to fan-planar graphs.
(In fact, the algorithm works for any embedded graph 
if we allow the order of crossings along an edge to change.)
A minor change allows to test the existence of {\em short} drawings instead.
For the smaller class of 1-planar graphs, yet another change allows
to test the existence of an {\em unconstrained} drawing.
All algorithms require crucially that the embedding is fixed.  

Recall that \Dujmovic et al.~\cite{DFK+08} gave an algorithm for this problem
for planar graphs where the embedding is not fixed; in the following we refer
to their algorithm as \DP.  The idea for our algorithm is to 
convert $G$ into a {\em planar} graph 
$G'$ such that $G$ has an embedding-preserving \ld{h} 
if and only if $G'$ has a plane \ld{h'} (where $h'=2h{-}1$).
One might be tempted to then appeal to \DP.
However, it is not at all clear whether \DP
could be modified to guarantee that the planar embedding is respected.
We therefore further modify $G'$ (in two steps) into a planar graph $G'''$ that has
a planar \ld{h'''} (where $h'''=12h'{+}1$) if and only if $G'$
has a plane \ld{h'}. Then call \DP on $G'''$.  

This latter step is of interest in its own
right: For plane graphs, we can test the existence of a {\em plane} 
\ld{h} in time FPT in $h$.  This improve on \DP, which permitted
changes of the embedding.

\iffull
To simplify the reductions, it is helpful to observe that \DP
allows further restrictions.  This algorithm first computes 
a path decomposition $\calP$ of small width.  
It then uses dynamic programming 
with table-entries indexed (among other things)
by the bags of $\calP$ and specifying
(among other properties) the layer for each vertex in the bag.
So we can impose restrictions on the layers that a vertex may be on.
Also, since for any edge some bag contains both endpoints, we can impose 
restrictions on the {\em span}, i.e., the distance between the layers of its 
endpoints.
We will impose even more complicated restrictions that require changing the
path decomposition a bit; this will be explained below.
\fi

\subsection{Proper drawings: Contracting Crossing Patches}
\label{sec:removing-crossings}

This section applies when we want to test the existence
of a {\em proper} \ld{h} (i.e., no long or flat edges are allowed). 
\iffull
We start with an easy lemma.  
\fi

\begin{lemma}
\label{lem:proper_crossings}
Let $G$ be an embedded graph with a crossing-patch $\calC$, and assume
$G$ has an embedding-preserving proper \ld{h} $\Gamma$. Then in the
embedding of $G_\calC$ induced by the one of $G$, 
all vertices of $V_\calC$ are on the 
infinite region.
\end{lemma}
\iffull
\begin{proof}
By Observation~\ref{obs:crossing_patch}, the induced drawing of subgraph $G_\calC$
lies entirely between two layers $L_i$ and $L_{i+1}$, with $V_\calC$
on these layers and hence on the infinite region. 
Since the drawing is embedding-preserving, $V_C$ hence
is on the infinite region of the
induced embedding of $G_\calC$.  
\end{proof}
\fi

Note that the conclusion of Lemma~\ref{lem:proper_crossings} depends only on the embedding of
$G$, not on $\Gamma$, and as such can be tested
given the embedding of $G$.  In the rest of this subsection we assume
that 
it
holds for all crossing-patches, as otherwise 
$G$ has no embedding-preserving proper layered drawing 
and we can stop.

As depicted in Fig.~\ref{fig:replace_crossings}, the operation of {\em contracting a crossing-patch $\calC$} consists of 
contracting all the edge-segments within $\calC$ to obtain one vertex $c$
that is adjacent to all of $V_\calC$.  
Hence, the rotation at $c$ lists
the vertices of $V_\calC$ in the order in which they appeared on the
infinite region of $G_\calC$.
As Fig.~\ref{fig:replace_crossings} suggests, we can
convert a proper layered drawing $\Gamma$ of $G$ into  a layered
drawing $\Gamma'$ of $G'$ with
roughly twice as many layers. To be able to undo such a conversion,
observe that $\Gamma'$ has special properties.   First, it is 
{\em 2-proper}, by which we mean that for any edge 
$(v,w)$ of $G$ the vertices
$v$ and $w$ are exactly two layers apart, and the edges incident to a
contracted vertex $c$ are proper.  It also {\em 
preserves monotonicity}: for any edge $(v,w)$ of $G$ that had a
crossing, the edges $(v,c)$ and $(c,w)$ are drawn such that their
union is a $y$-monotone curve.\footnote{As discussed later these properties can be tested within \DP.} 
Since $G'$ is obtained from $G$ by contracting crossing-patches,
and each contracted vertex $c$ can be placed at a dummy-layer between
the two layers surrounding the crossing-patches, one immediately verifies:

\begin{lemma}
\label{lem:replace_crossings_forward}
Let $G$ be an embedded graph, and let $G'$ be the
graph obtained by contracting crossing-patches. 
If $G$ has an embedding-preserving proper \ld{h} $\Gamma$ then
$G'$ has a plane monotonicity-preserving 2-proper \ld{(2h{-}1)}.
\end{lemma}

The other direction is not obviously true. It is easy to convert
a plane monotonicity-preserving 2-proper \ld{(2h{-}1)} of $G'$ to an \ld{h} of $G$ with the correct
rotation system and pairs of crossing edges
(the drawing is {\em weakly isomorphic} \cite{Schaefer18}).
But the {\em order} of crossings 
may change when connecting vertices by straight-line segments.
For example, in 
Fig.~\ref{fig:patch}(a),
moving the top left vertex much farther left would change the order of crossings
while keeping the rotation scheme unchanged.
So we give the other direction only for fan-planar graphs, where
this is impossible.%
\footnote{Another resolution would be to use polylines
between two layers, without requiring their bends to be on layers.  One can
argue that if $G$ had a straight-line embedding-preserving drawing, 
then such curves could be made $y$-monotone.}

\begin{lemma}
\label{lem:replace_crossings_backward}
Let $G$ be a fan-plane graph, and let $G'$ be the
graph obtained by contracting crossing-patches. 
If $G'$ has a plane monotonicity-preserving 2-proper \ld{(2h{-}1)} $\Gamma'$ 
then $G$ has a fan-plane proper \ld{h}. 
\end{lemma}
\iffull
\begin{proof}
Consider any crossing 
patch $\calC$ of $G$ that was contracted into vertex $c$,
say $c$ is on layer $L_i$ in $\Gamma'$.  Since the drawing
is 2-proper, all neighbours of $c$ are on $L_{i-1}$ or $L_{i+1}$.
Since for any edge $(v,w)$ in $E_\calC$ the endpoints are two layers apart,
therefore $v\in L_{i-1}$ and $w\in L_{i+1}$ or vice versa.  Remove the 
edges incident to $c$ and re-insert the edges in $E_\calC$ as straight-line
segments.

Since the rotation at
$c$ is respected, the order of $V_{\calC}$ on $L_{i-1}\cup L_{i+1}$ reflects the
order along the infinite region of $G_\calC$.
Two edges $e,e'$ in $E_{\calC}$ crossed in $G$ if and only if their
endpoints alternated in the order along the infinite region of $G_\calC$,
and so they cross in the resulting drawing as needed.

Assume an edge $e=(u,w)$ in $E_{\calC}$ crosses edges $e_1,\dots,e_d$ in $G$, 
in this order while walking from $u$ to $w$.  It suffices to argue that the
same order of crossings happens in the created drawing.
Let $v$ be the common endpoint of $e_1,\dots,e_d$, say $e_i=(v,w_i)$
for $i=1,\dots,d$.  We know that endpoints of $e,e_1,\dots,e_d$ are on 
the infinite region of $G_{\calC}$ since they belong to $V_{\calC}$.
Furthermore, their (clockwise or counter-clockwise) order along the 
infinite region must be exactly $v,u,w_1,\dots,w_d,w$ since we have
a good drawing. Namely, for any $i\in 1,\dots,d$ vertex $v$ must be
separated from $w_i$ in the order by $\{u,w\}$, otherwise
$e$ and $e_i$ would have to cross twice since they cross at least once.
Also, for any $i<j$, if the order along the infinite region is
$u,w_j,w_i,v$ while the order along $e$ is $u,e_i,e_j,v$, then
$e_j$ and $e_i$ would have to cross each other between where they cross $e$ and their
endpoints $w_i$ and $w_j$. 
In a good drawing no two edges cross twice and edges with
a common endpoint do not cross, so both are impossible.

Assume up to symmetry that $v\in L_{i-1}$, which means that $w_1,\dots,w_d$
are on $L_{i+1}$.  Since the rotation at $c$ contains $v,u,w_1,\dots,w_d,w$ 
in this order, $w_1,\dots,w_d$ are on layer $L_{i+1}$ in this order,
and edge $e$ crosses
$e_1,\dots,e_d$ in this order as desired.

Repeating this operation at all crossing patches hence gives a drawing of $G$
that respects the embedding.  After deleting
even-indexed layers (which contained no vertices of $G$), we obtain a
fan-plane proper \ld{h} of $G$.
\end{proof}
\fi
%
%

\subsection{Flat and long edges}

We will discuss in a moment how to test whether a graph has
a plane \ld{(2h{-}1)} that is monotonicity-preserving and 2-proper, 
but first study modifications
that allow us to test for short drawings (i.e., to allow flat edges)
and unconstrained drawings. 

Only minimal changes are needed when flat edges are allowed.  
Observation~\ref{obs:crossing_patch}, and therefore Lemma~\ref{lem:proper_crossings}
continue to hold.  
When there are no long edges, flat edges never have crossings. 
So it suffices to allow edges without crossings to have span 0 in $G'$.  
We say that a layered drawing $\Gamma'$ of $G'$ is {\em 2-short} if 
for any edge $(v,w)$ of $G$ the vertices
$v,w$ are either zero or two layers apart, and the edges incident to a
contracted vertex $c$ are proper.    

\begin{lemma}
\label{lem:replace_crossings_short}
Let $G$ be a fan-plane graph, and let $G'$ be the
graph obtained by contracting crossing-patches. 
$G$ has a fan-plane short \ld{h} if and only if 
$G'$ has a plane monotonicity-preserving 2-short \ld{(2h{-}1)}.
\end{lemma}
\iffull
\begin{proof}
The forward-direction is straightforward.  The backward direction
is proved almost exactly as in 
Lemma~\ref{lem:replace_crossings_backward},
except that preserving monotonicity is now vital (while it was not
actually needed in Lemma~\ref{lem:replace_crossings_backward}).
Namely, if $(v,w)$ is an edge involved in some crossing-patch
that was contracted to vertex $c$, then a 2-short drawing
permits $v$ and $w$ to be on the same layer, e.g. both above the
layer $L_i$ of $c$.  But monotonicity-preserving (and proper edges
incident to $c$) force them to be on layers $L_{i-1}$ and $L_{i+1}$
instead and the rest of the proof can proceed as before.
\end{proof}
\fi

Long edges pose difficulties because Observation~\ref{obs:crossing_patch}
no longer holds.  However, in a 1-plane graph $G$ every crossing-patch has
a single crossing, i.e., contracting crossing-patches is simply planarizing $G$.  
This crossing therefore either lies
between two layers or (if a long edge crosses a flat edge) exactly on a layer.
Define a drawing of $G'$ to be {\em 2-unconstrained} if every 
vertex of $G$ lies on an odd-indexed layer.  The following is shown almost exactly as
Lemma~\ref{lem:replace_crossings_forward}-\ref{lem:replace_crossings_short}; we leave the details to the
reader.

\begin{lemma}
\label{lem:1planePlanarization}
Let $G$ be a 1-plane graph and let $G'$ be its planarization.  
Then $G$ has a 1-plane unconstrained \ld{h} if and only if
$G'$ has a plane monotonicity-preserving 2-unconstrained \ld{(2h{-}1)}. 
\end{lemma}

\iffull
\subsection{Enforcing a rotation scheme}
\fi
\label{sec:enforcing-rotation}

Recall that we want a {\em plane} drawing of $G'$ while \DP tests the
existence of {\em planar} drawings.
As the next step we hence turn $G'$ into a graph $G''$ that is 
a subdivision of a 3-connected planar graph (hence has a
unique planar rotation scheme).  
There are many ways of making a planar graph 3-connected
(e.g. we could triangulate the graph or stellate every face), but we need to
use a technique here that allows to relate the height of layered
drawings of $G'$ and $G''$, and this seems hard when using triangulation
or stellation.  

Instead we use a different idea, which is easier to describe from the point of view of
{\em angles} of $G'$, i.e.,  incidences between a vertex
$v$ and a region $f$.  
\iffull
(A vertex may be incident to a face repeatedly, in case
of which this gives rise to multiple angles, but it should be
clear from the context which of them we mean.)
\fi
The operation of {\em filling the angles of $G'$} consists of two
steps.  
First, replace every edge $e$ of $G'$ by a {\em tripler-graph} $H$
\iffull
; $H$
consists of three (subdivided) copies of $e$ with some edges added to
make $H$ an inner triangulation 
\else
\fi
(see Fig.~\ref{fig:filler-paths}(b)).  
\iffull
Now add a {\em filler path} at every angle $v,f$ of $G'$ as follows.
Let $u,w$ be the clockwise/counter-clockwise
neighbour of $v$ on $f$ in $G'$.    Let $(v,u')$ and $(v,w')$ be the
edges of the tripler-graphs of $(v,u)$ and $(v,w')$ that are now
on $f$.
Add a subdivided edge between $u'$ and $w'$ and place it inside face $f$.  
\fi

\begin{figure}[tb]
\hspace*{\fill}
\subfigure[~]{\includegraphics[width=0.2\linewidth,page=3]{filler_path.pdf}}
\hspace*{\fill}
\subfigure[~]{\includegraphics[width=0.2\linewidth,page=1]{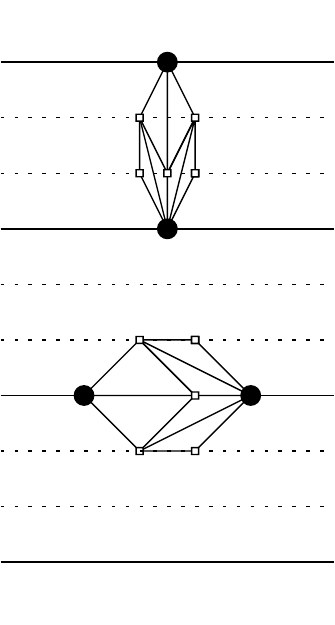}
\hspace*{-10mm}
\includegraphics[width=0.2\linewidth,page=7]{filler_path.pdf}}
\hspace*{\fill}
\subfigure[~]{\includegraphics[width=0.2\linewidth,page=4]{filler_path.pdf}}
\hspace*{\fill}
\hspace*{\fill}
\caption{(a) A very small plane graph $G'$. (b) Replacing a proper, flat or long edge with a tripler-graph. 
(c) Graph $G''$ obtained by filling all 
angles.    Some edges from the tripler-graphs are not shown.
}
\label{fig:doubleK2}
\label{fig:filler-paths}
\end{figure}

\iffull
\begin{lemma}
\label{lem:3connected}
Let $G'$ be a plane graph.  Let $G''$ be a graph obtained by filling the
angles of $G'$.  Then $G''$ is a subdivision of a 3-connected planar graph.
\end{lemma}
\begin{proof}  Let $G_c$ be the graph obtained from $G''$ by contracting
filler-paths into edges; we claim that $G_c$ is 3-connected.
We can view $G_c$ as having been built as follows:  Start with graph $G'$
and subdivide every edge.  For every face $f$ of degree $k$ of $G'$,
insert a cycle $C_f$ of length $2k$ inside $f$, and connect the vertices
of $C_f$ to their corresponding vertices on $f$.  
Add a few more edges connecting $C_f$ to $f$ such that all faces except
$C_f$ become triangles.  In particular, all faces of $G_c$ are simple
cycles, which immediately shows that $G_c$ is 2-connected.  Also note
that every vertex is incident to at most one non-triangular face, and
for every edge at most one endpoint is incident to a non-triangular face.
Now assume we have a cutting pair $\{v,w\}$, which means that at least
two faces contain both $v$ and $w$.  At least one of these faces must be
a triangle, which means that $(v,w)$ is an edge.  But then {\em all}
faces incident to both $v$ and $w$ must be triangles, by the above
condition on edges.  This is impossible if $\{v,w\}$ is a cutting pair
in a simple planar graph.
\end{proof}
\fi

Recall that we had some restrictions on drawings of $G'$, such as
being 2-proper and monotonicity-preserving.  All of them can be expressed 
as a {\em subgraph-restriction},
where we are given a (connected, constant-sized) subgraph $H$ of $G'$
and restrict the indices of layers used by $V(H)$.
\iffull
For example if $H$ is a single
vertex, then we can force its layer to be among a set of layers of our choice.  
If it is a single edge then we can force its span to be among a set of spans
of our choice.  If $H =v$-$c$-$w$ for some contracted vertex $c$ and edge
$(v,w)$ in $G$, then we can force $c$ to be within the range of the 
layers of $v,w$, hence $(v,c)\cup (c,w)$ is drawn $y$-monotonically. 
So this covers all the restrictions we had on $G'$.
We will discuss below how to test (under some assumptions) the
existence of a subgraph-restricted \ld{h} using \DP.

So assume graph $G'$ comes with some subgraph-restrictions $H_1,\dots,H_d$. 
For $j=1,\dots,d$, 
translate restriction $H_j$ to $G''$ by letting $H_j'$ be
graph $H_j$ with edges replaced by tripler-subgraphs, and layer-restrictions
replaced according to $i\leftrightarrow 3i{-}2$.     We impose further
subgraph-restrictions on $G''$:
(1) Every vertex of $G'$ of must be on a layer whose index is $2 \bmod 3$,
and (2) any tripler-graph $H$ must be drawn such that the {\em middle path}
(the path between vertices of $G'$ that uses no edges from the outer-face)
is drawn $y$-monotonically.
\fi

\begin{lemma}
\label{lem:fill_angles}
Let $G'$ be a plane graph.  Let $G''$ be a graph obtained by filling the
angles of $G'$.
Then $G'$ has a plane subgraph-restricted \ld{h} if and only if $G''$ has a
plane subgraph-restricted \ld{(3h)}. 
\end{lemma}
\iffull
\begin{proof}
Assume first that $G''$ has a plane subgraph-restricted \ld{(3h)} $\Gamma''$.  
The vertices of $G'$ occur only every third layer, and
the middle path of each tripler-graph is drawn $y$-monotonically.  Hence
after deleting filler-paths and tripler-graphs except for the middle paths,
we obtain a drawing of $G'$ on $h$ layers, with edges $y$-monotone since 
middle paths are $y$-monotone.  
The subgraph-restrictions of $G'$ are satisfied since
they were translated suitably into $G''$.

Now assume that $G'$ has a plane subgraph-restricted
\ld{h} $\Gamma'$, and insert a
dummy-layer before and after any layer of $\Gamma'$ to obtain $3h$ layers.
Insert tripler-graphs in place of their corresponding
edges using the appropriate drawing from Fig.~\ref{fig:filler-paths}.
Clearly all subgraph-restrictions are satisfied.
It remains to argue how to place filler-paths.  
Consider a face $f$ of $G'$ containing a path
$u$-$v$-$w$ in clockwise order; we filled the angle $v,f$ with filler-path
$u'$-$s$-$w'$ where $s$ is a degree 2 vertex.
Observe that $(v,u')$ and $(v,w')$ are
drawn proper, regardless of the chosen
drawing of the tripler-graphs.  
This puts $u'$ and $w'$ either on the same layer or two layers apart.
Walking from $u'$ to $w'$ along face $f$ hence requires at most one bend,
so we can draw the filler-path (with $s$ at the bend) such that all edges
are $y$-monotone.
See Fig.~\ref{fig:filler-paths}(c). 
\end{proof}
\fi

\iffull
\subsection{Enforcing the outer-face}
\label{sec:enforcing-outerface}

We do one more modification to enforce the outer-face.   Let $G''$
be a graph that is a subdivision of a 3-connected planar graph $G_c$;
we assume throughout that $G''$ is not a simple cycle since no simple
cycle would arise from the prior modifications.
The operation of {\em adding escape-paths} assumes that
we are given one face $f$ of $G''$ (the desired outer-face) and consists of
the following.  Add a new vertex $r$ inside $f$.  Pick three vertices
$z_1,z_2,z_3$ on face $f$ that were also vertices in the 3-connected graph
$G_c$; in particular $f$ is the {\em only} face that contains all three of
them.  Add three paths of length $n=|(V(G'')|$ 
that connect $z_1,z_2,z_3$ to $r$; we call these the
{\em escape-paths}.  

Graph $G''$ may have subgraph-restrictions, which we translate
to the resulting graph $G'''$ by changing layer-restrictions
according to $i\leftrightarrow 4i{-}2$.    We impose further subgraph-restrictions
on $G'''$:  Vertex $r$ is on the bottommost layer,
and any vertex of $G''$ is on $L_{4i-2}$ for some $i\geq 1$.

\fi
\iffull
\begin{lemma}
\label{lem:cage}
Let $G''$ be a planar graph that is a subdivision of a 3-connected graph,
embedded with face $f$ as the outer-face.
Let $G'''$ be
the graph obtained by adding escape-paths to $G''$.  Then 
$G''$ has a plane subgraph-restricted \ld{h} 
if and only if $G'''$ has a planar subgraph-restricted \ld{(4h{+}1)} .
\end{lemma}
\begin{proof}
If $G'''$ has a planar \ld{(4h{+}1)} $\Gamma'''$ that satisfies the restrictions, then
$r$ is on the bottommost layer, hence on the outer-face of $\Gamma'''$.   Remove $r$ and the 
escape-paths to get the induced drawing $\Gamma''$ of $G''$; this must have
$z_1,z_2,z_3$ on the outer-face since they are adjacent (via the escape-paths)
to $r$.  So the outer-face of $\Gamma''$ must be $f$.  The rotation scheme
of $G''$ is automatically respected since it is unique.
Finally vertices of $G''$ only on every fourth layer,
so by deleting
all other layers we get a plane \ld{h} of $G''$.
This satisfies the restrictions on $G''$ since they were inherited into $G'''$.

Vice versa, if $G''$ has a plane \ld{h} $\Gamma''$,
then insert three layers between any
two layers of $\Gamma''$, and also three layers above and one layer below 
$\Gamma''$.
Place $r$ in the topmost layer.  Clearly all subgraph-restrictions
are satisfied, except that we need to explain how to route the escape-paths.

Vertices $z_1,z_2,z_3$ are on the 
outer-face $f$ of $\Gamma''$, which also contains $r$.
Find, for $i=1,2,3$, a Euclidean shortest path
$\pi_i$ from $z_i$ to $r$ inside $f$.
(These three paths may overlap each other, but they do not cross.)
Now place the escape-paths by tracing near $\pi_i$, but using the 
nearest available layer inside face $f$ instead.   This is feasible, even
at a local minimum or maximum of $\pi_i$, since 
only every fourth layer of $\Gamma'''$ contains vertices of $G''$, and only those
vertices can be local minima/maxima.     
Therefore,
even if all three paths $\pi_1,\pi_2,\pi_3$ go through one local 
minimum/maximum, we can still use the three layers below/above it to 
place bends for the escape-paths.  See Fig.~\ref{fig:cage}.  These
layers have not been used for other bends of escape-paths already since
$\pi_i$ was a Euclidean shortest path.  

At each local minimum or maximum
the drawing $\Gamma_i$ of the escape-path to $z_i$ 
must use a vertex of degree 2 to ensure that edges are
drawn $y$-monotonically. There are at most $\deg(f)\leq n-1$ such
vertices, so there are sufficiently many degree-2 vertices in the
escape-paths.   If we did not use them all, then artificially add
more vertices at bends or insert flat edges to use them up.
See Fig.~\ref{fig:cage}.
Thus we can insert the escape-paths into the
drawing and obtain the desired planar proper \ld{(4h{+}3)} of $G'''$.
\end{proof}
\fi

\begin{figure}[tb]
\hspace*{\fill}
\includegraphics[width=0.3\linewidth,page=3]{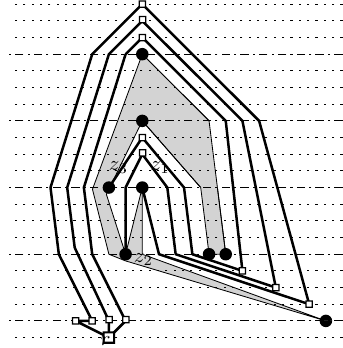}
\hspace*{\fill}
\caption{Routing the escape-paths (thick solid) along the outer-face.
\iffull For illustration purposes we chose paths that are longer than needed.\fi
}
\label{fig:cage}
\end{figure}

\subsection{Putting it all together}

\begin{theorem}
There are $O(f(h) poly(n))$ time algorithms of to test
the following:
\begin{itemize}
\item Given a fan-plane graph $G$, does it have a fan-plane proper \ld{h}?
\item Given a fan-plane graph $G$, does it have a fan-plane short \ld{h}?
\item Given a 1-plane graph $G$, does it have a 1-plane unconstrained \ld{h}?
\end{itemize}
\end{theorem}
\begin{proof}
First test whether the conclusion of Lemma~\ref{lem:proper_crossings}
is satisfied for all crossing-patches (this is trivially true for 
1-planar graphs).  If not, abort.  Otherwise
contract the crossing-patches of $G$ to obtain $G'$, and add the
subgraph-restrictions that $G'$ must be drawn monotonicity-preserving
and 2-proper/2-short/2-un\-con\-strained.  Fill the angles 
of $G'$ to obtain $G''$, and add escape paths to obtain $G'''$.
Inherit the above subgraph-restrictions into $G''$ and $G'''$.  
Also add the restrictions discussed when building $G''$ and $G'''$.
We have argued that $G'''$ contains a planar subgraph-restricted
\ld{(24h{-}11)} if and only if $G$ has the desired embedding-preserving \ld{h}.  

We can test for the existence of a planar \ld{(24h{-}11)} of $G'''$ using
\DP, the dynamic programming algorithm from \cite{DFK+08}.  (As this algorithm
is quite complicated, we will treat it as a black box and not review it here.)
As we argue now, in the same time we can also ensure the  created
subgraph-restrictions $H_1,\dots,H_d$.  Observe that every edge of $G'''$
belongs to a constant number of subgraph-restrictions,  and that each
$H_j$ has constant size.  Let $\calP$ be a path decomposition
of $G'''$ of width at most $24h$ (this must exist, otherwise $G'''$ has no 
\ld{(24h{-}11)}). $\calP$ is found as part of \DP.    Modify $\calP$ 
as follows:  For each $H_j$ that is not a single vertex,
and every bag $P$ that contains at least one edge of $H_j$,
add {\em all} vertices of $H_j$ to $P$.  The result $\calP'$
is a path decomposition since $H_j$ is connected.
Since bag $P$ represents $O(h)$ edges (it
induces a planar graph), and 
edges belong to constant number of restriction subgraphs of constant size, 
the bags
of $\calP'$ have size $O(h)$.   Call \DP on
$G'''$ using this path decomposition $\calP'$.  Since each table-entry
of the dynamic program specifies the layer-assignment, and since each
restriction subgraph $H_j$ appears in at least one bag $P$ of $\calP'$, 
we can enforce the subgraph-restriction by permitting (among the table-entries
indexed by bag $P$) only those that satisfy the restriction on $H_j$.  
\end{proof}

Sadly, our results are mostly of theoretical
interest.  Algorithm \DP is FPT in $h$, but the dependency $f(h)$
on $h$ is a very large function.
Our algorithm (where $h$ gets replaced by $24h$ and then increased
by another constant factor to accommodate the subgraph-restrictions) makes this
even larger.

\section{Testing Algorithm for $2$-Layer Fan-planarity}
\label{sec:trees}


Finally we turn to fan-planar drawings when the embedding is not fixed. 
\iffull
We have results here only for 2 layers (which are surprisingly complicated
already).
Graphs with maximal \propertwo drawings have been studied earlier by Binucci
et al.~\cite{DBLP:journals/jgaa/BinucciCDGKKMT17}. 
They characterized these graphs as subgraphs of a \stego
(illustrated in Fig.~\ref{fig:stegoEx}; we review its definition now).

A \emph{ladder} is a bipartite outer-planar graph consisting of two paths of the same length $\langle u_1, u_2, \dots, u_{\frac{n}{2}} \rangle$ and $\langle v_1, v_2, \dots, v_{\frac{n}{2}} \rangle$, called \emph{upper} and \emph{lower} paths, plus the edges $(u_i,v_i)$ $(i = 1,2, \dots, \frac{n}{2})$; the edges $(u_1,v_1)$ and $(u_{\frac{n}{2}},v_{\frac{n}{2}})$ are called the \emph{extremal edges} of the ladder. A \emph{snake} is a planar graph obtained from an outer-plane ladder, by adding, inside each internal face, an arbitrary number (possibly none) of paths of length two connecting a pair of non-adjacent vertices of the face. In other words, a snake is obtained by merging edges of a sequence of several $K_{2,h}$ ($h\geq2$). We may denote the partite set with more than 2 vertices (if any) the \emph{large side} of a $K_{2,h}$. A vertex of a snake is \emph{mergeable} if it is an end-vertex of an extremal edge and belongs to the large side of an original $K_{2,h}$. 
Mergeable vertices are black in Fig.~\ref{fig:stegoEx}.
A \emph{\stego} is a graph obtained by iteratively merging two snakes at a distinct mergeable vertex, and by adding degree-1 neighbors (``{\em stumps}'')
to mergeable vertices.   We call a vertex of degree 2 in a \stego a {\em joint vertex}; these are on the large side of a $K_{2,h}$.
Note that each ladder vertex has either three or four ladder vertices as neighbors, except for the vertices at extremal edges. If a ladder vertex has four neighboring ladder vertices, then it is a cut-vertex in $G$.
Binucci et al.~\cite{DBLP:journals/jgaa/BinucciCDGKKMT17}
showed that a graph is \propertwo if and only if it is a subgraph of a \stego.
Recognizing snakes (which are exactly the biconnected \propertwo graphs) is polynomial
\cite{DBLP:journals/jgaa/BinucciCDGKKMT17}, but
the complexity of recognizing \propertwo graphs that are not biconnected is open. 

Now we show how to test whether a tree has a \propertwo.
\fi

\begin{theorem}
Let $T$ be a tree with $n$ vertices. We can test in $O(n)$ time whether $T$ admits a \propertwo drawing.
\end{theorem}
\begin{proof}
\iffull
Suppose that $T$ is \propertwo and let $G$ be a \stego such that $T \subset G$. 
All stumps in $G$ are leaves in $T$.
Let $T'$ be the tree obtained from $T$ by removing all its leaves; then $T'$ contains no stumps.
We use the term {\em leafless} for a vertex of $T'$ that was not incident to any leaves of $T$.
We need some straightforward observations.

\begin{claim}
In $T'$, a ladder vertex $v$ of $G$ is adjacent to at most two joint vertices, and if it is incident
to two joint vertices $y,y'$ then they belong to distinct $K_{2,h}$. 
\end{claim}
\begin{proof}
If $v$ is incident to three joint vertices then two of them are in the same $K_{2,h}$, so we only need
to show the second claim.  Recall that $y,y'$ 
have degree 2 in $G$ and (since they are in $T'$) also have degree 2 in $T$.   
So the edges to their other common neighbour $z$ in the $K_{2,h}$ must also be in $T$, 
giving us a cycle $v_i$-$y$-$z$-$y'$ in $T$, an impossibility.  
\end{proof}

\begin{claim}
$T'$ contains no vertex with degree greater than four. 
\end{claim}
\begin{proof}
Assume for contradiction that (up to symmetry) $v_i$ has 5 neighbours in $T'$.  
We claim that nearly always the situation of the previous claim must happen and consider cases.  If $v_i$ was a mergeable  endvertex of a snake then 
it had at most 4 neighbours that are not stumps (hence potentially in $T'$).  If $v_i$ was an endvertex of a snake but not mergeable,
then it belongs to only one $K_{2,n}$ and has at most 2 neighbours of degree 3 or more, so the above situation applies.  So we are done
unless $v_i$ is in the middle of a snake, and its two incident $K_{2,h}$-subgraphs (which share edge $(v_i,u_i)$ in $G$) contain exactly one 
joint-vertex each while the other three neighbours of $v_i$ in $T'$ are $v_{i-1},u_i$ and $v_{i+1}$.  
Call the two joint-vertices $x_{i-1}$ and $x_{i+1}$ (connected to $u_{i-1}$ and $u_{i+1}$).  As before, edges $u_{i-1}$-$x_{i-1}$-$v_i$-$x_{i+1}$-$u_{i+1}$
all must exist in $T$.  But then $u_i$ can {\em only} be connected to $v_i$ in $T$, because its only other incident edge $(u_i,u_{i-1})$ and $(u_i,u_{i+1})$
would lead to a cycle in $T$.  Therefore $u_i$ has degree 1 in $T$ and is not in $T'$, a contradiction.
\end{proof}

\medskip

The idea now is to test whether $T'$ (and hence $T$) can be augmented to a \stego $G$ (without stumps). 
Let $\Pi=\langle u_1,u_2,\dots,u_l \rangle$ be the longest path of $T'$. 
The vertices of $\Pi$ with degree greater than two represent ladder vertices of $G$ whose subtrees must be ``paired'' with corresponding sub-paths in $\Pi$ in order to reconstruct $G$.

We assign a \emph{type} to each node $u_i$ of $\Pi$ as follows. 
If $u_i$ has degree four, it is of \emph{type A}.
If $u_i$, with $1 < i <l$, has degree two and is leafless, and both $u_{i-1}$ and $u_{i+1}$ have degree three, then $u_{i-1}$, $u_i$, and $u_{i+1}$ are of \emph{type B} and form a \emph{triple}.
If $u_i$ has degree three (and it is not of type B), it is of \emph{type C}.
If $u_i$ has degree two (and it is not of type B), it is of \emph{type D}.
Call a subtree of $T'$ \emph{primary} if it contains a vertex of $\Pi$, and {\em secondary} otherwise.

\begin{figure}[ht]
\centering
\includegraphics[width=0.6\linewidth,page=3]{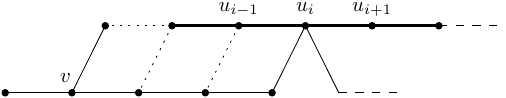}
\caption{Showing the longest path of $T'$ (bold) and the case in which a secondary subtree rooted at a type-A vertex $u_i$ contains a vertex of degree greater than two. The dotted edges are those missing to reconstruct a putative ladder containing the secondary subtree.}
\label{fig:type}
\end{figure}

\begin{claim}
If $u_i$ is of type $A$, then no secondary subtree rooted at $u_i$ contains a vertex of degree 3 in $T'$.
\end{claim}
\begin{proof}
If $u_i$ is of type A, there are two primary and two secondary subtrees rooted at $u_i$. Observe that $u_i$ is a cut-vertex vertex of $G$, that is, there are two snakes in $G$ that were merged at $u_i$. 
Then each secondary subtree of $u_i$ is a path that belongs to one of these two snakes. To see this, recall that $T'$ contains no stumps of $G$ and that each ladder vertex of $G$  is adjacent to at most one joint vertex in $T'$. Hence, a vertex $v$ of degree greater than two is adjacent to either three ladder vertices or two ladder vertices and one joint. In both cases, vertex $v$ cannot belong to a secondary subtree as it would contain a path longer than the longest path in one of the two primary subtrees, see Fig.~\ref{fig:type} for an illustration. 
\end{proof}

If $u_{i-1}$, $u_i$, and $u_{i+1}$ are a triple of type B, with a similar argument one can prove that the secondary subtree rooted at $u_{i-1}$ and that the secondary subtree at $u_{i+1}$ are both paths.

If $u_i$ is of type C, then its secondary subtree may contain at most one vertex $v$ of degree greater than two, namely, such a vertex (if any) has degree three and it is either directly adjacent to $u_i$ or there is one degree-2 node between $u_i$ and $v$. 
	Again, this is implied by the fact that a vertex of degree greater than two is a ladder vertex and that a secondary subtree cannot contain a sub-path whose length is longer than  the longest path in one of the two primary subtrees of $u_i$. 

If $u_i$ is of type D, there is no secondary subtree rooted at $u_i$.

The idea is now to use a greedy strategy along the nodes of $\Pi$ starting from $u_1$ in order to reconstruct a \stego $G$ containing $T'$. Recall that a \stego is composed of distinct snakes, each having an underlying ladder. Hence, we can view all upper (lower) paths of these ladders as a unique upper (lower) path that goes through extremal edges.
Without loss of generality, we assign $u_1$ to be the first ladder vertex of the upper path of $G$.
While we extend the path $\Pi$ from left to right on either the upper or lower path, we keep track of an integer offset variable that stores the minimum required number of ladder vertices on the respective \emph{other} ladder path. 
This offset can take positive and negative values, depending on whether the other path extends beyond the currently considered node $u_i$ of $\Pi$ or not.
For negative offsets the considered primary node extends further to the right than the last node on the other ladder path, while for positive offsets the primary path lags behind.
An offset of $0$ means that both ladder paths extend equally far.
Let $o_i$ be this offset value corresponding to node $u_i$.

Furthermore we need to store a flag $\beta_i$ with each offset $o_i$ that expresses whether ($\beta_i = 1$) or not ($\beta_i = 0$) at least one leafless node exists on the ladder path opposite of $\Pi$ in the currently extended snake. 
By default each flag is set to 0.
This information is required in some extremal cases for deciding if all leaves can be re-inserted in the end of the process or not.
In fact, the existence of a single leafless node on a ladder in which \emph{all} ladder vertices are contained in $T'$ guarantees that all other nodes of $T'$ on this ladder can have arbitrarily many leaves. 
However, the existence of a leafless node on the ladder path is only beneficial if the choice is between two options with identical offsets. Otherwise the smaller offset option will always allow for strictly more freedom when continuing to extend the primary path $\Pi$.

For each node $u_i$ of $\Pi$, we distinguish all possible cases based on its type. 
Node triples of type-B are considered together in one step. 
We assume that the primary subtree that contains $u_{i-1}$ 
has already been processed by the algorithm. 
Without loss of generality we can assume that $u_{i-1}$ is part of the upper path of $G$; otherwise we simply flip the roles of the upper and lower path.
While, depending on the type of $u_i$, there are many options of distributing the primary subtree containing $u_{i+1}$ and the secondary subtrees, we observe that 
(i) it is sufficient to consider the primary path going monotonically from left to right and 
(ii) if several options are valid it is sufficient to select the one yielding the smallest offset $o_i$, and in case of ties one with $\beta_i = 1$. This is because the smaller an offset and the more leafless nodes on the secondary path, the more freedom we maintain for placing future secondary paths. Additionally, in a \stego, any decision for a type-A, -B, or -C node $u_i$ only depends on the offset $o_{i-1}$ and the flag $\beta_{i-1}$ of its predecessor $u_{i-1}$ on $\Pi$ and the length of its secondary subtree(s). Hence we can greedily select among all feasible options the one producing smallest offset $o_i$ and as a tie-breaker a larger value $\beta_i$. For ease of presentation, we discuss the types in the order A,C,D, and finally B.

\textbf{Type A.} If $u_i$ is of type A, we have four options of embedding the primary subtree containing $u_{i+1}$ and the two secondary subtrees, see Fig.~\ref{fig:node-types-A}. Recall that $u_i$ is a cut node of $G$ and thus there are two snakes (and hence ladders) that meet at $u_i$. 
Let $s_1$ and $s_2$ be the two secondary subtrees, which are in fact paths. 
Let $|s_1|$ and $|s_2|$ denote their lengths and assume $|s_1| \le |s_2|$. 
If $o_{i-1} + |s_2| - 1 \le 0$ then there is enough space in the ladder of $u_{i-1}$ to assign $s_2$ (or $s_1$) to be part of its lower path. 
In this case it is better to assign $s_2$ to the left ladder as this yields the smaller offset $o_i$ (compare Fig.~\ref{fig:node-types-A}(a) and (c) or Fig.~\ref{fig:node-types-A}(b) and (d)).
If the condition on the length of $s_2$ is satisfied with equality, we need to additionally check that $\beta_{i-1} = 1$ or that the secondary path embedded to the left has a leafless node -- else we have to reject the instance.
Otherwise, if $o_{i-1} + |s_2| - 1 > 0$ but $o_{i-1} + |s_1| - 1 \le 0$ then $s_1$ must be part of the left ladder (see Fig.~\ref{fig:node-types-A}(a--b)) and we again need to verify the existence of a leafless node on the ladder path in case of equality.
In either case, we obtain a smaller offset $o_i$ if the primary path stays on the upper path of $G$ (compare Fig.~\ref{fig:node-types-A}(a--b) or Fig.~\ref{fig:node-types-A}(c--d)).
Thus we set $o_i = |s_1| - 1$ if $s_1$ is embedded to the right, or $o_i = |s_2| - 1$ if $s_2$ is embedded to the right.
We set $\beta_i = 1$ if the path embedded to the right has a leafless node.
If, however, $o_{i-1} + |s_1| -1 > 0$ then we cannot embed either secondary path into the left ladder and thus report that $T'$ is no subgraph of a \stego $G$.

\begin{figure}[tb]
	\centering
		\includegraphics[scale=1,page=1]{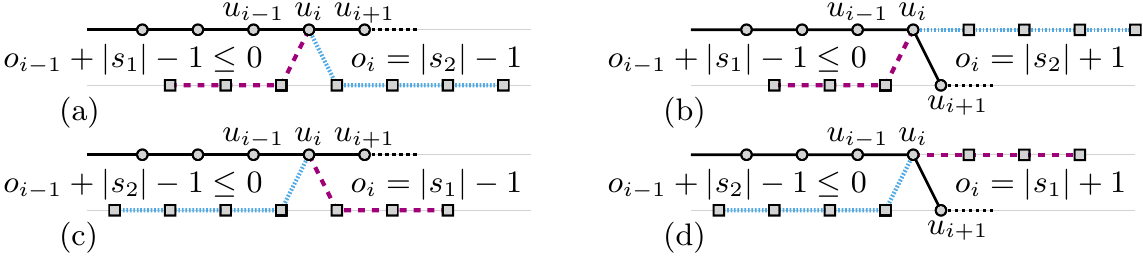}
	\caption{Different possibilities of embedding $T'$ for a node $u_i$ of type A.}
	\label{fig:node-types-A}
\end{figure}

%
%

\textbf{Type C.} If $u_i$ is of type C we distinguish two sub-cases depending on whether the secondary tree $s$ rooted at $u_i$ is a path or not.
We first consider that $s$ is a path, see Fig.~\ref{fig:node-types-C}(a--d).
In the general case that any vertex of $s$ may have leaves in $T$ and thus needs to be embedded as a ladder vertex we have two options (Fig.~\ref{fig:node-types-C}(a--b)).
If $o_{i-1} + |s| - 1\le 0$ then we can embed $s$ into the left ladder (modulo the existence of a leafless node in case of equality) and set $o_i = - 1$ and $\beta_i = 0$ (Fig.~\ref{fig:node-types-C}(a)).
Otherwise, if $o_{i-1} \le 1$ then $s$ can at least be embedded into the lower path of the right ladder and we set $o_i = |s| - 1$ and $\beta_i = 1$ if $s$ has a leafless node (Fig.~\ref{fig:node-types-C}(b)).
Finally, if $o_{i-1} > 1$ then $s$ cannot be embedded into either ladder and we report that $T'$ is no subgraph of a \stego $G$.
In the special case that the neighbor of $u_i$ in $s$ is leafless, it can in fact be embedded as a joint vertex of $G$, which gives us two additional options shortening the required length of $s$ in the ladder by two. The conditions and resulting offsets are given in Fig.~\ref{fig:node-types-C}(c--d), where, if possible, (d) is strictly preferred over (b) and (a) is strictly preferred over (c). Notice that in case of equality in Fig.~\ref{fig:node-types-C}(c) the leafless node adjacent to $u_i$ cannot be used a second time; thus another leafless node must exist in order to consider Fig.~\ref{fig:node-types-C}(c) a valid option. Likewise, in case Fig.~\ref{fig:node-types-C}(d) we can only set $\beta_i = 1$ if a second leafless node exists in $s$.

\begin{figure}[tb]
	\centering
		\includegraphics[scale=1,page=3]{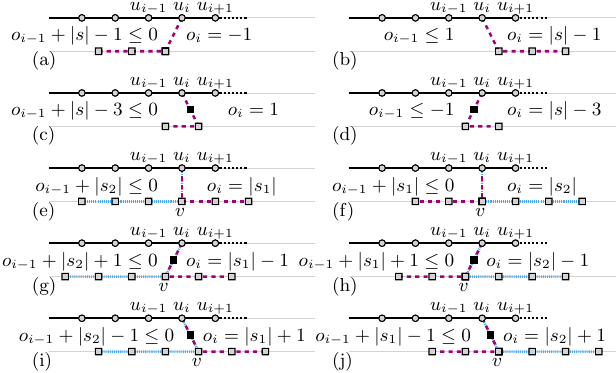}
	\caption{Different possibilities of embedding $T'$ for a node $u_i$ of type C. Nodes marked as black squares must be leafless as they are mapped to a joint vertex of~$G$.}
	\label{fig:node-types-C}
\end{figure}

In the second case, the secondary tree $s$ contains exactly one degree-3 node, $v$, which is either the immediate neighbor of $u_i$, see Fig.~\ref{fig:node-types-C}(e--f), or there is a single leafless degree-2 node between $v$ and $u_i$, see Fig.~\ref{fig:node-types-C}(g--j). 
In both cases, $v$ is the root of two branches of $s$, which are in fact paths that we denote by $s_1$ and $s_2$.
Let us assume that $|s_1| \le |s_2|$ and that $v$ is a neighbor of $u_i$.
If $o_{i-1} + |s_2| \le 0$ then we can assign the longer path $s_2$ to the left part of the ladder (modulo the existence of a leafless node in case of equality) and set $o_i = |s_1|$ and, if there is a leafless node in $s_1$, $\beta_i = 1$ (Fig.~\ref{fig:node-types-C}(e)).
Else if $o_{i-1} + |s_1| \le 0$ we assign $s_1$ to the left (modulo the existence of a leafless node in case of equality) and set $o_i = |s_2|$ and, if $s_2$ has a leafless node, $\beta_i = 1$ (Fig.~\ref{fig:node-types-C}(f)).
If $o_{i-1} + |s_1| > 0$ then $s$ cannot be embedded on the lower path of $G$ at all and we report that $T'$ is no subgraph of a \stego $G$.

If there is a node $w$ between $u_i$ and $v$, we need to map $w$ to a joint vertex of $G$. 
We have two options of arranging $s_1$ and $s_2$ on the left and right part of the ladder and two options of positioning $w$, see Fig.~\ref{fig:node-types-C}(g--j).
Among these four combinations, we again pick the one that satisfies the constraints for the left side and minimizes the offset $o_i$ (if there is a tie and one of the paths contains a leafless node, we pick the option that yields $\beta_i = 1$).

\textbf{Type D.} If $u_i$ is of type D we simply extend $\Pi$ and assign $u_i$ to be part of the upper path just as its predecessor $u_{i-1}$. The new offset $o_i$ is obtained by decreasing the previous offset $o_{i-1}$ by one, i.e., $o_i = o_{i-1} - 1$. The existence of a leafless node on the ladder path is also simply inherited, i.e., we set $\beta_i = \beta_{i-1}$.

\begin{figure}[tb]
	\centering
		\includegraphics[scale=1,page=2]{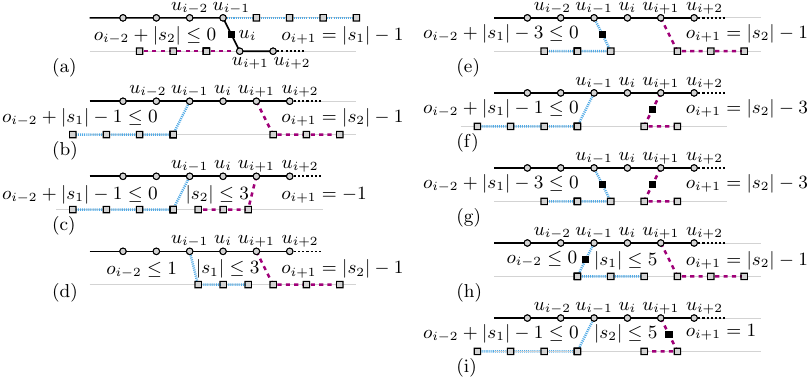}
	\caption{Different possibilities of embedding $T'$ for a node triple $u_{i-1}, u_i, u_{i+1}$ of type~B. Nodes marked as black squares must be leafless as they are mapped to a joint vertex of $G$.}
	\label{fig:node-types-B}
\end{figure}

\textbf{Type B.} Triples of type B are a special case to consider. Let $u_{i-1}, u_i, u_{i+1}$ be the triple vertices, and let $s_1$, $s_2$ be the secondary paths rooted in $u_{i-1}$, $u_{i+1}$, respectively. In this triple the leafless degree-two vertex $u_i$ can be mapped to a joint vertex of $G$, which allows a local backward flip of the secondary paths $s_1$  and $s_2$.
Figure~\ref{fig:node-types-B}(a) shows the only configuration of the triple, in which path $s_2$ is mapped to the left part of the ladder and path $s_1$ to the right part, despite $u_{i+1}$ being two positions right of $u_{i-1}$ in $\Pi$.
This is possible if $o_{i-2} + |s_2| \le 0$ (modulo the existence of a leafless node in case of equality) and yields a new offset of $o_{i+1} = |s_1| - 1$ and flag $\beta_{i+1} = 1$ if $s_1$ contains a leafless node.
Moreover, this is the only case, in which $\Pi$ changes from the upper path of $G$ to the lower path (or vice versa).
If, on the other hand, we do not apply the backward flip, then the  embedding options of the type-B triple are the same as treating it as a sequence of the type-C node $u_{i-1}$, the type-D node $u_i$, and the type-C node $u_{i+1}$. As a consequence, all of them embed the secondary path $s_1$ left of $s_2$. For completeness, Fig.~\ref{fig:node-types-B}(b--i) illustrate all relevant combinations. The details have been already discussed for types C and D.

\bigskip

We finally reintroduce the degree-one vertices that we removed when going from $T$ to $T'$. We claim that this is always possible and that the leaves of $T$ can always be viewed as stumps or joints of a \stego $G$. So far we mapped each vertex of $T'$ as either a ladder or a joint vertex, and in the latter case such a vertex has no leaves in $T$. Moreover, we know for each ladder vertex whether it belongs to the upper path or to the lower path of a ladder. With this mapping of vertices to ladders, our goal is to reinsert the missing edges and vertices that form a \stego $G$ containing $T$ as a (not necessarily spanning) subgraph. In particular, since $T'$ does not contain cycles, all edges that connect a vertex on the upper path of a ladder to the corresponding vertex in the lower path are missing. Moreover, there is also a non-empty sub-path missing in either the upper path or the lower path of each ladder, say the upper path, which leaves some freedom in the reconstruction of $G$. We will exploit this freedom for the reinsertion of the leaves of $T$. Consider a ladder $L$ underlying a snake $N$ of $G$ and consider the sub-path $S$ of $L$ that is not in $T'$. Assume first that $S$ contains at least one vertex, then this vertex does not belong to $T$, as we assume the leaves of $T$ be stumps or joints. We reinsert $S$ and we draw the edges that connect opposite ladder vertices of $L$ so to fully reconstruct $L$. See Fig.~\ref{fig:leaves-1} for an illustration where the edges not in $T$ are dashed and the vertex not in $T$ is larger. In order to reinsert the leaves of $T$ that are adjacent to vertices of $L$, consider first the two mergeable vertices of $N$. Their leaves can be reinserted as stumps. Consider now the leftmost and the rightmost cells of $L$, i.e., those that contain the two mergeable vertices. If they coincide, then there is only one vertex of $L$ whose leaves need to be reinserted, and we can reinsert them as joint vertices that belong to this cell. If they don't coincide, then we assign the leaves as shown in Fig.~\ref{fig:leaves-2}, where the leaves are solid disks (the figure shows the case in which the two mergeable vertices of $N$ are on the same path of $L$, the case in which are on opposite paths is similar). 

\begin{figure}[t]
\centering
\subfigure[~]{\includegraphics[width=0.48\linewidth,page=1]{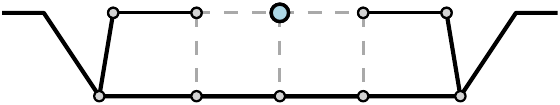}\label{fig:leaves-1}}\hfill
\subfigure[~]{\includegraphics[width=0.48\linewidth,page=2]{tree-leaves.pdf}\label{fig:leaves-2}}
\subfigure[~]{\includegraphics[width=0.48\linewidth,page=3]{tree-leaves.pdf}\label{fig:leaves-3}}\hfill
\subfigure[~]{\includegraphics[width=0.48\linewidth,page=4]{tree-leaves.pdf}\label{fig:leaves-4}}
\hfil
\caption{Illustration for the leaves reinsertion step.}
\label{fig:stego}
\end{figure}

Suppose now that $S$ contains just one edge, call it $(a,b)$. Add $(a,b)$ and draw the edges that connect opposite ladder vertices as to fully reconstruct $L$, see Fig.~\ref{fig:leaves-3}. We call the cell of $L$ containing $(a,b)$ \emph{central}. The previous greedy algorithm guarantees that all the leaves of $T$ adjacent to vertices that belong to the cells to left or to the right of the central one can be reinserted without using the central cell. This is due to the fact that in one of the two corresponding secondary paths there is a leafless node and hence the leaves can be reinserted as in the previous case (see  Fig.~\ref{fig:leaves-4} where the leafless node is large). Thus we can use the freedom given by the central cell to assign the  leaves of the other secondary path, as shown in Fig.~\ref{fig:leaves-4} where the leafless node is on the left path and the right path uses the central cell.


The above algorithm works in linear time: It first constructs $T'$ from $T$ by removing $O(n)$ leaves; it then traverses $\Pi$ and makes a constant number of operations for each node of $\Pi$; it finally reinserts the removed leaves by reconstructing a \stego $G$ that contains $T$ and has size $O(n)$. 
\fi
\end{proof}

\section{Summary and future directions}

We studied layered drawings of fan-planar graphs.
Motivated by the algorithm by \Dujmovic et al.~\cite{DFK+08}, and using
it as a subroutine, we gave an algorithm that tests the existence
of a fan-plane proper \ld{h} and is fixed-parameter tractable in $h$.
(Variation can handle fan-plane short or 1-plane unconstrained drawings.)
For the situation where the embedding of the graph is not fixed, we studied
the existence of fan-planar proper \lds{2} for trees.  Along the way, we also bounded the
pathwidth of graphs that have a fan-planar (short or proper) \ld{h}, and argued that such graphs have
a bar-1-visibility representation.
Many open problems remain:

\begin{itemize}
\item Are there FPT algorithms to test whether a 
graph has a fan-planar \ld{h} for $h>2$, presuming we can change the
embedding?  This problem was non-trivial even for trees and $h=2$
and proper drawings.

\item Our FPT algorithm for 1-plane unconstrained \ld{h}
permits bends on the long edges.  Is there an algorithm that tests
for the existence of 1-plane {\em straight-line} \ld{h}?   We could
easily test for a 1-plane $y$-monotone \ld{h}, but in contrast
to planar drawings, not all such drawings can be ``stretched'' to make
edges straight-line.

\item Our FPT algorithm for fan-plane drawings only worked for proper
or short drawings.
Is there an FPT algorithm if long edges are allowed?

\item Likewise, our pathwidth-bounds 
work only for proper or short \lds{h}.  Does every graph with an
fan-planar unrestricted \lds{h} have pathwidth $O(h)$?  Does it have a 
bar-1-visibility representation?
Note that fan-planar graphs
are not closed under subdividing edges, so we cannot simply replace long
edges by paths.

\iffull
\item The dynamic programming algorithm by \Dujmovic et al.~\cite{DFK+08}~is 
quite involved, and in particular, does {\em not} appeal to Courcelle's
theorem that states that any problem expressible in monadic second-order
logic is solvable in polynomial time on graphs of bounded pathwidth
\cite{Courcelle90}.  This raises the natural question:  Can we express
whether a graph has a \ld{k} (perhaps under some restrictions such as
proper or fan-planar) 
in monadic second-order logic?  
\fi
\end{itemize}

Last but not least, what do we know about layered drawings of $k$-planar
graphs for $k>1$?  Note that these are not necessarily fan-planar.

\bibliography{fanplanar}

\begin{thebibliography}{10}

\bibitem{Bie14}
T.~Biedl.
\newblock Height-preserving transformations of planar graph drawings.
\newblock In C.~Duncan and A.~Symvonis, editors, {\em Graph Drawing (GD'14)},
  volume 8871 of {\em LNCS}, pages 380--391. Springer, 2014.

\bibitem{DBLP:journals/jgaa/BinucciCDGKKMT17}
C.~Binucci, M.~Chimani, W.~Didimo, M.~Gronemann, K.~Klein,
  J.~Kratochv{\'{\i}}l, F.~Montecchiani, and I.~G. Tollis.
\newblock Algorithms and characterizations for 2-layer fan-planarity: From
  caterpillar to stegosaurus.
\newblock {\em J. Graph Algorithms Appl.}, 21(1):81--102, 2017.

\bibitem{Courcelle90}
B.~Courcelle.
\newblock The monadic second-order logic of graphs. {I}. {R}ecognizable sets of
  finite graphs.
\newblock {\em Inf. Comput.}, 85(1):12--75, 1990.

\bibitem{DBLP:journals/csur/DidimoLM19}
W.~Didimo, G.~Liotta, and F.~Montecchiani.
\newblock A survey on graph drawing beyond planarity.
\newblock {\em {ACM} Comput. Surv.}, 52(1):4:1--4:37, 2019.

\bibitem{DFK+08}
V.~Dujmovic, M.~Fellows, M.~Kitching, G.~Liotta, C.~McCartin, N.~Nishimura,
  P.~Ragde, F.~Rosamond, S.~Whitesides, and D.~Wood.
\newblock On the parameterized complexity of layered graph drawing.
\newblock {\em Algorithmica}, 52:267--292, 2008.

\bibitem{FLW03}
S.~Felsner, G.~Liotta, and S.~Wismath.
\newblock Straight-line drawings on restricted integer grids in two and three
  dimensions.
\newblock {\em J. Graph Alg. Appl}, 7(4):335--362, 2003.

\bibitem{HR92}
L.~Heath and A.~Rosenberg.
\newblock Laying out graphs using queues.
\newblock {\em SIAM J. Comput.}, 21(5):927--958, 1992.

\bibitem{KaufmannU14}
M.~Kaufmann and T.~Ueckerdt.
\newblock The density of fan-planar graphs.
\newblock {\em CoRR}, abs/1403.6184, 2014.

\bibitem{DBLP:journals/csr/KobourovLM17}
S.~G. Kobourov, G.~Liotta, and F.~Montecchiani.
\newblock An annotated bibliography on 1-planarity.
\newblock {\em Computer Science Review}, 25:49--67, 2017.

\bibitem{Schaefer18}
M.~Schaefer.
\newblock {\em Crossing Numbers of Graphs}.
\newblock CRC Press, 2018.

\bibitem{Sud04}
M.~Suderman.
\newblock Pathwidth and layered drawings of trees.
\newblock {\em Intl. J. Comp. Geom. Appl.}, 14(3):203--225, 2004.

\end{thebibliography}
\bibliographystyle{abbrv}

\end{document}